\documentclass[fleqn]{llncs}
\usepackage{etex}

\usepackage{graphicx}
\usepackage{pgfplots}
\pgfplotsset{width=12.5cm, height=8cm}

\usepackage{array}
\usepackage{datetime}
\usepackage{paralist}
\usepackage{caption}
\usepackage{subcaption}
\usepackage{bbm}
\usepackage{pstricks}
\usepackage{colortbl}
\usepackage{enumitem}
\usepackage{color,soul}
\usetikzlibrary{patterns}
\usepgfplotslibrary{groupplots}
\usetikzlibrary{patterns}
\usepackage{color}
\usepackage{pst-node}
\usepackage{pst-grad}
\usepackage{pst-coil}
\usepackage{pst-text}
\usepackage{pst-rel-points}
\usepackage{pstricks-add}
\usepackage{amssymb}
\usepackage{amsmath}
\usepackage{mathptm}
\usepackage{fancyvrb}
\usepackage{calc}
\usepackage{etex}
\usepackage{ifthen}
\usepackage{multirow}
\usepackage{rotating}
\usepackage{longtable}
\usepackage[normalem]{ulem}
\usepackage{stmaryrd}
\usepackage{hyperref}
\usepackage{xspace}
\usepackage{mathtools}
\usepackage{datetime}
\usepackage{algorithm}
\usepackage[noend]{algpseudocode}

\usepackage{xr}
\usepackage{bm}

\usepackage{tikz}
\usepackage{listings}
\usepackage{color}
\newdateformat{titledate}{\THEDAY\ \monthname[\THEMONTH], \THEYEAR}
\newdateformat{mydate}{\monthname[\THEMONTH] \THEYEAR}

\protect{\newlength{\BRACEL}}
\protect{\newlength{\TAL}} 
\newlength{\codeLineLength}

\newcommand{\transform}{\text{\em transform\/}\xspace}

\newcommand{\init}{\text{\em init\/}\xspace}

\newcommand{\cskw}{\text{\bf \{\/}\xspace}
\newcommand{\cekw}{\text{\bf \}\/}\xspace}
\newcommand{\qkw}{\text{\bf ?\/}\xspace}
\newcommand{\ckw}{\text{\bf :\/}\xspace}
\newcommand{\ifkw}{\text{\bf if\/}\xspace}
\newcommand{\forkw}{\text{\bf for\/}\xspace}

\newcommand{\breakkw}{\text{\bf break\/}\xspace}

\newcommand{\continuekw}{\text{\bf continue\/}\xspace}
\newcommand{\assertkw}{\text{\bf assert\/}\xspace}

\newcommand{\elsekw}{\text{\bf else\/}\xspace}

\newcommand{\arrvar}{\text{$a$}\xspace}
\newcommand{\absind}{\text{$i_a$}\xspace}
\newcommand{\absarr}{\text{$x_a$}\xspace}

\newcommand{\ms}{\text{$\sigma$}\xspace}

\newcommand{\msl}{\text{$\sigma_l$}\xspace}
\newcommand{\msm}{\text{$\sigma_m$}\xspace}
\newcommand{\msb}{\text{$\sigma_b$}\xspace}
\newcommand{\msA}{\text{$\sigma_a$}\xspace}
\newcommand{\expr}{\text{$e$}\xspace}
\newcommand{\sexpr}{\text{$\mathbbm{E}$}\xspace}
\newcommand{\var}{\text{$x$}\xspace}

\newcommand{\size}{\text{\sf\em lastof}\xspace}
\newcommand{\loc}{\text{$\ell$}\xspace}
\newcommand{\svar}{\text{$\mathbb{V}$}\xspace}

\newcommand{\earr}{\text{$\mathbb{E}_A$}\xspace}

\newcommand{\eval}[2]{\text{$\left\llbracket #1 \right\rrbracket_{#2}$}\xspace}
\newcommand{\sval}{\text{$\mathbb{C}$}\xspace}
\newcommand{\vexpr}{\text{\eval{e}{\ms}}\xspace}

\newcommand{\asvexpr}{\text{\eval{e}{\smsa}}\xspace}
\newcommand{\pc}{\text{$P$}\xspace}
\newcommand{\pa}{\text{$P'$}\xspace}
\newcommand{\msa}{\text{$\sigma'$}\xspace}
\newcommand{\smsa}{\text{$\Sigma'$}\xspace}
\newcommand{\smsalp}{\text{$\Sigma'\!\!_{l'}$}\xspace}
\newcommand{\smsamp}{\text{$\Sigma'\!\!_{m'}$}\xspace}

\newcommand{\msalp}{\text{$\sigma'\!\!_{l'}$}\xspace}
\newcommand{\msamp}{\text{$\sigma'\!\!_{m'}$}\xspace}

\newcommand{\msaap}{\text{$\sigma'\!\!_{a'}$}\xspace}
\newcommand{\msabp}{\text{$\sigma'\!\!_{b'}$}\xspace}

\renewcommand{\int}{\text{\sf I}\xspace}

\newcommand{\smem}{\text{$\mathbb{M}$}\xspace}
\newcommand{\smema}{\text{$\mathbb{M'}$}\xspace}

\newcommand{\assert}{\text{$A_n$}\xspace}
\newcommand{\assertinv}{\text{$A\!_n^{inv}$}\xspace}

\newcommand{\emit}{\text{\sf\em emit}\xspace}
\newcommand{\fullarrayaccess}{\text{\sf\em fullarrayaccess}\xspace}
\newcommand{\loopbound}{\text{\sf\em loopbound}\xspace}
\newcommand{\loopdefs}{\text{\sf\em loopdefs}\xspace}
\newcommand{\T}{\text{\sf\em transform}\xspace}

\newcommand{\Prog}{\text{\sf\em P}\xspace}
\newcommand{\Stmt}{\text{\sf\em S}\xspace}
\newcommand{\Expr}{\text{\sf\em E}\xspace}
\newcommand{\Lvalue}{\text{\sf\em L}\xspace}
\newcommand{\Init}{\text{\sf\em I}\xspace}

\begin{document}


\title{Scaling Bounded Model Checking By Transforming Programs With Arrays}

\author{Anushri Jana \inst{1} \and Uday P. Khedker\inst{2} \and Advaita Datar \inst{1} \and R. Venkatesh\inst{1} \and Niyas C \inst{1}
}
\institute{ Tata Research Development and Design Centre, Pune, India\\
\email{\{anushri.jana,advaita.datar,r.venky,niyas.c\}@tcs.com}
\and Indian Institute of Technology Bombay, India. \\
\email{uday@cse.iitb.ac.in}
}

\maketitle
\begin{abstract}
Bounded Model Checking is one the most successful techniques for finding bugs in program. 
However, model checkers are resource hungry and are often unable to verify
 programs with loops iterating over large
arrays. We present a transformation that enables bounded model checkers to
verify a certain class of array properties. Our technique transforms an
array-manipulating (\textsc{Ansi}-C) program to an array-free and loop-free (\textsc{Ansi}-C)
program thereby reducing the resource requirements of a model checker significantly. 
Model checking of the transformed program using an off-the-shelf bounded model checker
simulates the loop iterations efficiently.
Thus, our transformed program is a sound abstraction of the original program and
is also precise in a large number of cases---we formally characterize the 
class of programs for which it is guaranteed to be precise.
 We demonstrate the
applicability and usefulness of our technique on both industry code as well as
academic benchmarks.
\end{abstract}


\section{Introduction} 
\label{introduction} 

Bounded Model Checking is one of the most successful techniques for
finding bugs~\cite{copty2001benefits} as evidenced by success 
achieved by the tools implementing this technique in verification
competitions~\cite{SVCOMP15,SVCOMP16}. Given a program P
and a property $\varphi$, Bounded Model Checkers (BMCs) unroll the loops in P, a fixed
number of times and search for violations to $\varphi$ in the unrolled program.
However, for programs with loops of large or unknown bounds, bounded model
checking instances often exceed the limits of resources available. In our
experience, programs manipulating large arrays invariably have such loops
iterating over indices of the array. Consequently, BMCs
routinely face the issue of scalability in proving properties on arrays. The
situation is no different even when the property is an {\it array invariant}
i.e., it holds for every element of the array---a characteristic which can
potentially be exploited for efficient bounded model checking.

Consider the example in Figure~\ref{Mex1} manipulating an array of structures
$a$. The structure has two fields, $p$ and $q$, whose values are assigned in
the first {\it for} loop (lines 8--13) such that $a[i].q$ is the square of
$a[i].p$ for every index $i$. The second {\it for} loop (lines 14--17) asserts
that this property indeed holds for each element in $a$. This is a {\it
safe} program i.e., the assertion does not admit a counterexample.
CBMC~\cite{cbmc}, a bounded model checker for C, tries to unwind the first
loop 100000 times and runs out of memory before reaching
the loop with the assertion. We tried this example with several 
model checkers\footnote{Result for \emph{motivatingExample.c} at \htmladdnormallink{https://sites.google.com/site/datastructureabstraction/}{https://sites.google.com/site/datastructureabstraction/}} 
	and none of them
was able to prove this property because of a large loop bound.

\begin{figure*}[t]
\scriptsize
\begin{minipage}{0.45\textwidth}
\begin{lstlisting}[frame=single,language=c,escapeinside={(*@}{@*)},basicstyle=\ttfamily]
1. struct S {
2.    unsigned int p;		  
3.    unsigned int q;		 		  
4. } a[100000];
5. int i,k;

06. main()
07. {
08.  for(i=0; i<100000; i++)
09.  {
10.   k = i; 
11.   a[i].p = k;
12.   a[i].q = k * k ;
13.  }

14.  for (i=0; i<100000; i++)
15.  {
16.   (*@ \color[rgb]{0.75,0.164,0.164}{assert}@*)(a[i].q == 
          a[i].p * a[i].p);
17.  }
18. }
\end{lstlisting}
\caption{Motivating Example}
\label{Mex1}
\end{minipage}
\quad
\quad
\begin{minipage}{0.50\textwidth}
\begin{lstlisting}[frame=single,language=c,escapeinside={(*@}{@*)},basicstyle=\ttfamily]
1. struct S{
2.    unsigned int p;		  
3.    unsigned int q;		 		  
4. }x_a;
5. int i_a;
6. int i,k;
7. main()
8. {
9.   i_a = nd(0,99999);

//first loop body
10.  k = nd(0,100000);
11.  i = i_a;
12.  k = i;
13.  (i == i_a)? x_a.p = k : k;
14.  (i == i_a)? x_a.q = k * k : k*k ;
15.  k = nd(0,100000);
  
//second loop body
16.  i = i_a;
17.  (*@ \color[rgb]{0.75,0.164,0.164}{assert}@*)(((i==i_a)?x_a.q:nd())
        ==((i==i_a)?x_a.p:nd())
        *((i==i_a)?x_a.p:nd()));
18. }

\end{lstlisting} 
\caption{Transformed Code}
\label{TEx1}
\end{minipage}
\end{figure*}

One of the ways of proving this example safe is to show that the property holds
for any arbitrary element of the array, say at index $i_c$. This allows us to
get rid of those parts of the program that do not update $a[i_c]$ which, in
turn, eliminates the loop iterating over all the array indices. This enables
CBMC to verify the assertion without getting stuck in the loop unrolling.
Moreover, since $i_c$ is chosen nondeterministically from the indices of $a$,
the property holds for every array element without loss of generality.

This paper presents the transformation sketched above with the aim that the
transformed program is easier for a BMC to verify as compared
to the original program. The transformation
is over-approximative i.e.,  it give more values than that by the original program.
This ensures that if the original program is safe with respect to the chosen
property, so is the transformed program. However, the over-approximation
raises two important questions spanning practical and intellectual considerations:
\begin{enumerate}
\item {\it Is the proposed approach practically useful? Does the
transformation enable a BMC to verify real-world programs or
academic benchmarks fairly often?}

We answer this through an extensive experimental evaluation over
industry code as well as 
examples in the array category of SV-COMP
2016 benchmarks. 
Our approach helps CBMC to scale in each case.
We further demonstrate the applicability of
our technique to successfully identify a large number of false warnings (on an average 73\%)
reported by a static analyzer on arrays in large programs.

\item {\it Is it possible to characterize a class of properties for which it is
precise?}

In order to address this we provide a formal characterization of properties for
which the transformation is precise i.e., we state criteria under which
the transformed program is unsafe only when the original program is unsafe
(Section~\ref{sec::prec}).
\end{enumerate}

To summarize, this paper makes the following contributions:
\begin{itemize}
\item A new technique combining static analysis and model checking enabling verification of array invariant
	properties in programs with loops iterating over large arrays.
\item A novel concept of using pair of a \emph{witness variable} and a \emph{witness index} which allows us to remove the
loops and arrays from programs and simulate the iterations accessing different elements of arrays during model checking.
\item A formal characterization of properties for which the transformation is
precise.
\item A transformation engine implementing the technique.
\item An extensive experimental evaluation showing the applicability of our
technique to real-world programs as well as academic benchmarks.
\end{itemize}
		
The rest of the paper starts with an informal description of the transformation
(Section~\ref{desc}) before we define the semantics
(Section~\ref{sec:dsa}) and formally state the transformations rules
(Section~\ref{transformation}).  Section~\ref{sec:proof} and ~\ref{sec::prec},
respectively, describe the soundness and precision of our approach.  Section~\ref{exp}
presents the experimental setup and results. We discuss the related work in
Section~\ref{relwork} before concluding in Section~\ref{conclusion}.

\section{Informal Description}
\label{desc}

Given a program \pc containing loops iterating over an array \arrvar, we transform it 
to a program \pa that has a pair \text{$\langle \absarr, \absind \rangle$} of a
\emph{witness variable} and a \emph{witness index} for the array and its index
such that
$\absarr$ represents the element $a[i_a]$ of the original program. Further, loops are replaced by 
their customized bodies that operate only on the witness variable $\absarr$ instead of all elements of the array $a$.

To understand the intuition behind our transformation, consider a trace $t$ of \pc ending on the assertion
\assert. Consider the last occurrence of a statement \text{$s: a[e_1]=e_2$} in $t$. 
We wish to transform \pc such that
there exists a trace $t'$ of \pa ending on \assert in which the value of \absind is equal to that of $e_1$
and the value of \absarr\ is equal to that of $e_2$.  
%
%
We achieve this by transforming the program such that: 
\begin{itemize}
	\item \absind gets a non-deterministically chosen value at the start of the program (this facilitates an arbitrary choice of array element $a[\absind]$).
    	\item array writes and reads for $a[\absind]$ are replaced by the witness variable \absarr.   
	\item array writes other than $a[\absind]$ are eliminated and reads are replaced by a non-deterministically chosen 
               value.
       \item loop body is executed only once either unconditionally or non-deterministically; based on loop characteristics. 
       During the execution of the transformed loop body, 
	       \begin{itemize}
		\item the loop iterator variable either gets the value of \absind or a non-deterministically chosen value (depending on loop characteristics), and
		\item all other scalar variables whose values may be different in different iterations 
              		gets non-deterministically chosen values.
	\end{itemize}
\end{itemize}

Figure~\ref{TEx1} shows the transformed program \pa for the program \pc of Figure~\ref{Mex1}. 
Function \texttt{nd(l,u)} returns a non-deterministically chosen value in the range $[l..u]$.
In \pa, the witness index $i\_a$ for array $a$ is globally assigned a 
non-deterministically chosen value within the range of array size (at line 9).
In a run of BMC, the assertion is checked for this non-deterministically chosen element $a[\absind]$.
To ensure that the values for the same index $a[i_a]$ are written and read, we replace the
array accesses by the witness variable \emph{\text{$x\_a$}} only when the value of the index $i$ matches with that of $i\_a$ (lines 13, 14 and 17).
We remove the loop header but retain the loop body. 
To over-approximate the effect of the removal of loop iterations, we add non-deterministic assignments to all variables modified in the loop body, 
at the start of the transformed loop body and also
after the transformed loop body (lines 11 and 15). 
Note that we retain the original assignment statements too (line 12).
Since the loops at line 8 and line 14 in the original program iterate over the entire array, we equate loop iterator variable $i$ to
$i\_a$ (line 11 and 16) and the transformed loop bodies 
(lines 10--14 and lines 16--17) are executed unconditionally. 

 We are using a single variable, \emph{\text{$x\_a$}}, to represent the array \emph{a}. 
	However, $x\_a$ takes different values
	in different runs of the BMC based on an arbitrarily chosen value
	of $i\_a$ and \emph{k}. 
	Our technique is able to verify the program in Figure~\ref{Mex1} because we do not conflate the values of expressions 
	across different runs as is done in any static analysis. 
	In a sense, we go beyond
	a static analysis and use a dynamic analysis through a BMC to verify the property.

We explain the transformation rules formally in Section~\ref{transformation}.
The transformed program can be verified by an off-the-shelf BMC. 
Note that each index will be considered in some run of the BMC since $i\_a$ is chosen non-deterministically.
Hence, if an assertion fails for any index in the original program, it fails in the transformed program too.

\section{Semantics}
\label{sec:dsa}
In this section we formalize our technique
by explaining the language and defining a representation of states. 

\subsection{Language}
We formulate our analysis over a language modelled on C.
For simplicity of exposition we restrict our description to a subset of C which includes structures and 1-dimensional arrays.
For a given program,
let \sval, \svar, and \sexpr be the sets of values computed by the program, the variables appearing in the program, and the expressions
appearing in the program, respectively.
A value \text{$c \in \sval$} can be an integer, a floating-point or a boolean value.
A variable \text{$v \in \svar$}
can be a scalar variable, a structure variable, or an array variable. 
We define our program to have only one array variable denoted as \arrvar. 
However, our implementation handles multiple arrays as explained in 
our technical report~\cite{Jana2016Scaling}. 
We also define \text{$\earr \subseteq \sexpr$} 
as the set of array expressions of the form \text{\arrvar[\Expr]}.
An lval \Lvalue can be an array access expression or a variable. 
Let \text{$c \in \sval$}, 
\text{$\var,i \in (\svar - \{\arrvar\})$}.
We consider assignment statements, conditional statements, loop statements, and assertion statements defined by the following grammar. 
We define the grammar of our language using the following non-terminals:
Program \Prog consists of statements \Stmt which may use 
lvalues \Lvalue and expressions \Expr. We assume that programs are type correct as per C typing rules.
\begin{equation}
\renewcommand{\arraystretch}{1.2}
\begin{array}{rl}
\Prog \rightarrow & \; \,
	\Stmt
	\\
\Stmt  \rightarrow & \; \,
	\ifkw\; (\Expr) \; \Stmt \; \elsekw \; \Stmt 
	\;\; \big\lvert \;\; \ifkw\; (\Expr) \; \Stmt 
	\;\; \big\lvert \;\; \forkw \; (i = \Expr \,;\; \Expr \,;\; \Expr) \;\Stmt 
	\;\; \big\lvert 
	\\
	&
	\;\; \Stmt \; ; \; \Stmt 
	\;\; \big\lvert \;\; \Lvalue = \Expr 
	\;\; \big\lvert \;\; \assertkw(\Expr) 
	\\
\Lvalue \rightarrow & \;\,
	a[\Expr] 
	\;\; \big\lvert \;\; \var
	\\
		\label{eq:grammar.org}
\Expr \rightarrow & \; \,
	\Expr \oplus \Expr
	\;\; \big\lvert \;\; \Lvalue
	\;\; \big\lvert \;\; c
\end{array}
\end{equation}


In practice, we analyze \textsc{Ansi}-C language programs that include functions, pointers, composite data-structures, 
all kinds of definitions, and all control structures except multi-dimensional arrays.

\subsection{Representing Program States}
We define program states in terms of memory locations and the values stored in memory locations.
We distinguish between \emph{atomic} variables (such as
scalar and structure variables) whose values can be copied atomically using a single assignment operation, from
non-atomic variables such as arrays. Since we are considering 1-dimensional arrays, the
array elements are atomic locations. 

Function \text{$\loc(a[i])$} returns the memory location corresponding to the $i^{th}$ index of array $a$.
The memory of an input program consists of all atomic locations:
\begin{align}
	\smem & = (\svar - \{\arrvar\}) \cup \big\{ \loc(\arrvar[i]) \; \big\lvert \; 0 \leq i \leq \size(\arrvar) \big\}
\label{equation:mem.org}
\end{align}
The function \size(\arrvar) returns the highest index value for array \arrvar.

A \emph{program state} is a map $\ms : \smem \to \sval$. 
\vexpr denotes the value of expression \expr in the program state \ms. 

We transform a program by creating a pair \text{$\langle \absind,\absarr \rangle$} for array \arrvar where \absind is the witness index 
and \absarr is the witness variable.
The memory of a transformed program with additional variables is:
\begin{align}
	\smema & = (\svar - \{\arrvar\}) \cup \{\absarr\} \cup \; \{\absind\}
\label{equation:mem.trans}
\end{align}

For a transformed program, a program state is denoted by \msa and is defined over \smema. 
A set of states is denoted by \smsa and set of states at program point $l$ as 
$\smsa_l$. \asvexpr denotes the set of value of expression \expr in the set of program states \smsa.

We illustrate the states in the original and the transformed program through an example.
Let a program \pc have an array variable $a$ and a variable $k$ holding the size of the array $a$. Let
the array contain the values \text{$c_i \in \sval$}, \text{$0 \leq i < n$}, where $n \in \sval$ is the value of size of the array.
Then, a program state, $\ms$ at any program point \emph{l} can be: 
\begin{align}
\ms &= 
	\big\{  \big(k,n\big), 
         \big(\loc \left(a[0]\right),c_0\big), 
         \big(\loc \left(a[1]\right),c_1\big),
         \ldots, 
         \big(\loc \left(a[{n-1}]\right),c_{n-1} \big)
 \big\}
	\label{eq:state.p}
\end{align}

In the transformed program \pa, let \absarr and \absind be the witness variable and witness index respectively. 
Let $l'$ be the program point in \pa that corresponds to $l$ in \pc.
Then, all possible states in the transformed program at $l'$ are,

\begin{align}
\msa\!\!_{0} & = \{(k,n),(i_a,0),(x_a,c_{0})\}
	\nonumber
\\
\msa\!\!_{1} & = \{(k,n),(i_a,1),(x_a,c_{1})\}
	\nonumber
\\
\ldots
	\nonumber
\\
\msa\!\!_{{n-1}} & = \{(k,n),(i_a,{n-1}),(x_a,c_{{n-1}})\}
	\nonumber
\\
\smsa & = \{ \msa\!\!_{0},\msa\!\!_{1} \ldots \msa\!\!_{{n-1}}\}
	\nonumber	
\end{align}


We now formally define how the states at a program point in the transformed program represents a state at the
corresponding program point in the original program. 

\begin{definition}
\label{def:abselement1}
Let \ms be a state
at a program point in \pc and let \msa be a state at the corresponding program point in \pa.
Then, for any index $c_{1}$ under consideration, \msa represents \ms, denoted as \text{$\msa \rightsquigarrow \ms$}, if \\
$\big(\loc(a[c_1]),c_2\big) \in \ms$ ~~$\Rightarrow$~~
$\msa = \left\{
		(\absind,c_1), 
  		(\absarr,c_2) 
	\right\} 
	\;\cup\;
	\left\{
		(y,c) \mid (y,c) \in \ms, 
   		y \in (\svar-\{\arrvar\})
	\right\}$
\end{definition}

\begin{definition}
\label{def:abselement2}
Let \ms be a state
at a program point in \pc and let \smsa be set of states at the corresponding program point in \pa.
Then, \smsa represents \ms, denoted as \text{$\smsa \rightsquigarrow \ms$}, if 

$\forall c_{i}$ such that $\big(\loc(a[c_i]),c_j\big) \in \ms$, $\exists \msa \in \smsa$ such that $\msa \rightsquigarrow \ms$.
\end{definition}

    



Let \assert be the assertion at line n in program P. Let \ms be a state reaching \assert in the original program with pair
\text{$\left(\loc\left(a\left[\eval{e_1}{\ms}\right]\right),\eval{e_2}{\ms}\right)$}. 
Let \msa be the state in 
transformed program, \msa represents \ms. Thus, \msa has two pairs,  
\text{$\left(\absind, \eval{e_3}{\msa}\right)$} and
\text{$\left(\absarr,\eval{e_4}{\msa}\right)$}
such that
\text{$\eval{e_3}{\msa} = \eval{e_1}{\ms}$} and \text{$\eval{e_4}{\msa} = \eval{e_2}{\ms}$} .
Hence, if the assertion \assert holds in transformed program, it holds in the original program too.



\begin{figure}[!t]
\addtocounter{figure}{1}
\newcommand{\DEFRule}{%
\arrayrulecolor{lightgray}\hline\arrayrulecolor{black}
}
\begin{align*}
\T(\Expr)  & = 
	\\
	&\hspace*{-10mm}
		\Expr \equiv (\Expr_1 \oplus \Expr_2)
		\hspace*{-5mm}&\Rightarrow\hspace*{2mm}
		&
		\emit\left(\T(\Expr_1) \; \oplus \; \T(\Expr_2) \right)
		\tag{\thefigure.$E_1$}\label{transform.expr.1}
		\\
	&\hspace*{-10mm}
		\Expr \in \earr, \;
		\Expr \equiv \arrvar[\Expr_1] \;
		\hspace*{-5mm}&\Rightarrow\hspace*{2mm}
		&
		\emit\left(\left(\Expr_1 ==\absind\right)  \qkw \; \absarr \; \ckw \; nd() \right)
		\tag{\thefigure.$E_2$}\label{transform.expr.2}
		\\
	&\hspace*{-10mm}
		\text{otherwise}
		\hspace*{-5mm}&\Rightarrow\hspace*{2mm}
		& \; \emit\left(\Expr   \right)
		\tag{\thefigure.$E_3$}\label{transform.expr.3}
	\rule[-.5em]{0em}{1em}
	\\ 
	\hline
	\rule{0em}{1em}
\T(\Stmt)  & = 
	\\
	&\hspace*{-10mm}
		\Stmt \equiv (\Lvalue=\Expr), \; 
		\Lvalue \equiv \arrvar[\Expr_1] \;
		\hspace*{-5mm}&\Rightarrow\hspace*{2mm}
		&
		\emit\big(\left(\Expr_1 ==\absind\right) \qkw \; 
			\\
	& & & \phantom{\emit\big(} \absarr 
			=  \T(\Expr) \; \ckw \; \T(\Expr) \big)
		\tag{\thefigure.$S_1$}\label{transform.stmt.1}
		\\
	&\hspace*{-10mm}
		\Stmt \equiv (\Lvalue=\Expr),
		\Lvalue \not\equiv a[\Expr_1] \;
		\hspace*{-5mm}&\Rightarrow\hspace*{2mm}
		& \;\emit\left(\Lvalue=\T(\Expr)\right)
		\tag{\thefigure.$S_2$}\label{transform.stmt.2}
		\\
	&\hspace*{-10mm}
		\begin{array}[t]{@{}l@{\,}l}
		\Stmt \equiv & (\forkw(i = \Expr_1;\;\Expr_2;\;\Expr_3) \; \Stmt_1), 
		\\
		& 
		\fullarrayaccess(\Stmt), 
		\\
		& u \in \loopdefs(\Stmt_1)
		\end{array}
		\hspace*{-4.5mm}&\Rightarrow\hspace*{2mm}
		&
		\;\begin{array}[t]{@{}l@{\,}l}
		\emit\;\big( &  u=nd();
			\hspace*{8mm} // \forall u \in \loopdefs(\Stmt_1) \\
		& i=\absind;\\ 
			 & \T(\Stmt_1);\\  
			& u=nd(); 
			\hspace*{8mm} // \forall u \in \loopdefs(\Stmt_1) \\ 
		\phantom{\emit}	\; \big)
		\end{array}
		\tag{\thefigure.$S_3$}\label{transform.stmt.3}
		\\
	&\hspace*{-10mm}
		\begin{array}[t]{@{}l@{\,}l}
		\Stmt \equiv & (\forkw(i = \Expr_1;\;\Expr_2;\;\Expr_3) \; \Stmt_1),  
		\\
		& 
		\neg\fullarrayaccess(\Stmt),
		\\
		& u \in \loopdefs(\Stmt_1)
		\end{array}
		\hspace*{-5mm}&\Rightarrow\hspace*{2mm}
		&
		\;\begin{array}[t]{@{}l@{\,}l}
		\emit\;\big( & \ifkw(nd(0,1))  \\ 
		& \cskw  \; \; u=nd();
		 	\hspace*{5mm} // \forall u \in \loopdefs(\Stmt_1) \\ 
		& \;\; \;  i=nd(\loopbound(\Stmt));\\
			& \;\;\; \T(\Stmt_1); \\ 
		& \cekw \\ 
		& \; u=nd();
			\hspace*{7mm} // \forall u \in \loopdefs(\Stmt_1) \\ 
		 \phantom{\emit}\;\big)
		\end{array}
		\tag{\thefigure.$S_4$}\label{transform.stmt.4}
		\\
	&\hspace*{-10mm}
		\Stmt \equiv (\ifkw(\Expr) \; \Stmt_1 \; \elsekw \; \Stmt_2)
		\hspace*{-5mm}&\Rightarrow\hspace*{2mm}
		& \;\emit\big(\ifkw(\T(\Expr)) 
			\\
		& & & \phantom{\emit\big(} \; \T(\Stmt_1) \; \elsekw \; \T(\Stmt_2) \big)
		\tag{\thefigure.$S_5$}\label{transform.stmt.5}
		\\
	&\hspace*{-10mm}
		\Stmt \equiv (\ifkw(\Expr) \; \Stmt_1 )
		\hspace*{-5mm}&\Rightarrow\hspace*{2mm}
		& \;\emit\left(\ifkw(\T(\Expr)) \; \T(\Stmt_1)  \right)
		\tag{\thefigure.$S_6$}\label{transform.stmt.6}
		\\
	 &\hspace*{-10mm}
		\Stmt \equiv (\Stmt_1;\Stmt_2)
		\hspace*{-5mm}&\Rightarrow\hspace*{2mm}
		& \;\emit\left(\T(\Stmt_1);\T(\Stmt_2) \right)
		\tag{\thefigure.$S_7$}\label{transform.stmt.7}
		\\
	&\hspace*{-10mm}
		\Stmt \equiv (\assertkw(\Expr))
		\hspace*{-5mm}&\Rightarrow\hspace*{2mm}
		& \;\emit\left(\assertkw(\T(\Expr)) \right)
		\tag{\thefigure.$S_8$}\label{transform.stmt.8}
		\\
	&\hspace*{-10mm}
		\text{otherwise}
		\hspace*{-5mm}&\Rightarrow\hspace*{2mm}
		& \; \emit\left(\Stmt\right)
		\tag{\thefigure.$S_9$}\label{transform.stmt.9}
	\rule[-.5em]{0em}{1em}
	\\ 
	\hline
	\rule{0em}{1em}
\T(\Prog)  & = 
	\\
	&\hspace*{-10mm}
		\Prog \equiv \Stmt 
		\hspace*{-5mm}&\Rightarrow\hspace*{2mm}
		& 
		\;\begin{array}[t]{@{}l@{\,}l}
		\emit\;\big( & 
			\absind = nd\left(\size\left(a\right)\right);
				\\
			& \T(\Stmt)
			\\
		 \phantom{\emit}\;\big)
		\end{array}
		\tag{\thefigure.$P$}\label{transform.prog.1}
\end{align*}



\addtocounter{figure}{-1}
\caption{Program transformation rules. Non-terminals \Prog, \Stmt, \Expr, \Lvalue represent the code fragment in the input program
derivable from them.}
\label{fig:transform}
\end{figure}

\section{Transformation}
\label{transformation}

The transformation rules are given in Figure~\ref{fig:transform}. A transformed program
satisfies the following grammar 
derived from that of the original program (Grammar \ref{eq:grammar.org}).
Let \text{$x,\absarr,\absind \in \svar$} denote a scalar variable, the witness variable, and the witness index, respectively.
Let \text{$c,l,u \in \sval$} denote the values. Then,
\begin{equation}
\renewcommand{\arraystretch}{1.2}
\begin{array}{rl}
\Prog \rightarrow & \; \,
	\Init \;; \; \Stmt
	\\
\Init \rightarrow & \; \,
	\absind = nd (l,u) 
	\\
\Stmt  \rightarrow & \; \,
	\ifkw\; (\Expr) \; \Stmt \; \elsekw \; \Stmt 
	\;\; \big\lvert \;\; \ifkw\; (\Expr) \; \Stmt 
	\;\; \big\lvert \;\; \Stmt \; ; \; \Stmt 
	\;\; \big\lvert \;\; 
	\Lvalue = \Expr 
	\;\; \big\lvert \;\; \assertkw(\Expr) 
	\\
\Lvalue \rightarrow & \;\,
	\var \;\; \big\lvert \;\; \absarr \;\; \big\lvert \;\; \absind
	\\
\Expr \rightarrow & \; \,
	\Expr \oplus \Expr
	\;\; \big\lvert \;\; \Lvalue
	\;\; \big\lvert \;\; c
	\;\; \big\lvert \;\; nd()
	\;\; \big\lvert \;\; nd(l,u)
\end{array}
\label{eq:grammar.trans}
\end{equation}
The non-terminal \Init represents the initialization statements for witness index.
Witness variable is initialized in the same scope as that in the original program.

Following are the functions used in the transformation rules.
\begin{compactitem}
\item Function $nd(l,u)$ returns a non-deterministically chosen value 
	between the lower 
      limit $l$ and the upper limit $u$. When a range is not provided, $nd$ 
	returns a non-deterministically chosen value based on the type of \Lvalue. 
\item Function \T transforms the 
      code represented by its argument non-terminal. Function \emit shows the actual
      code that would be emitted. 
      We assume that it takes the code emitted by \T and possibly some additional statements and 
	outputs the combined code.
      It has been used only to distinguish the transformation time activity and run time activity.
      For example, the boolean conditions in cases~\ref{transform.expr.2} and~\ref{transform.stmt.1} are not evaluated by the
      body of function \T but is a part of the transformed code and is evaluated at run time when the transformed program is executed.
	Similar remarks apply to the \ifkw statements and other operations inside the parenthesis of \emit function.
\item Function $\fullarrayaccess(\Stmt)$ analyzes\footnote{\label{note1}Results of analysis may be over-approximated.} the characteristics of the loop \Stmt.
	\begin{itemize}
		\item When the loop \Stmt accesses array \arrvar completely, $\fullarrayaccess(\Stmt)$ returns true.
		This means that loop either reads from or writes to 
		all the indices of the array; this could well be under different conditions in the code.
	\item When the loop \Stmt accesses array \arrvar partially, $\fullarrayaccess(\Stmt)$ returns false. This means
	that the loop may not access all the indices or some indices are being read while some other indices are being written.
	\item When loop \Stmt does not access an array, $\fullarrayaccess(\Stmt)$ returns false. 
	\end{itemize}


\item Function $\loopdefs(\Stmt)$ returns a possibly over-approximated set of variables modified in loop \Stmt.
	\begin{compactitem}	
	\item Scalar variables are included in this set if they appear on the left hand side of an assignment statement in \Stmt 
	(except when the RHS is a constant). 
	\item Loop iterator variable $i$ of loop \Stmt is not included in this set.
	\item 
	Array variable \arrvar 
	is included in this set when the array access expression appears on the left hand side of an assignment
	and the value of index expression is different from the current value of the loop iterator $i$.
\end{compactitem}

\item Function $\size(a)$ returns the highest index value for array \arrvar. 
\end{compactitem}
With the above functions, the transformation rules are easy to understand. Here we explain non-trivial transformations.
\begin{compactitem}
\item To choose an array index for a run, witness index (\absind) is initialized at the start of the program 
	to a non-deterministically chosen value from the range of the indices of the array (case~\ref{transform.prog.1}). 
	This value determines the array element ($a[\absind]$) represented by the witness variable (\absarr). 
\item An array access expression in LHS or RHS is replaced by the witness variable (\absarr) only when the
      values of the witness index and index expression of the array access expression match.
      When the values do not match, it implies that an element other than \absind\ is being accessed.
      Hence for any other index the assignment does not happen (case~\ref{transform.stmt.1}). 
      However RHS of such assignment statement is retained to handle side-effects of 
      expressions (not included in the grammar). 
      Similarly, when an index other than \absind is read, it is
      replaced by a non-deterministically chosen value (case~\ref{transform.expr.2}). 
\item Loop iterations are eliminated by removing the loop header containing initialization, test, and increment expression for 
	the loop iterator variable. The loop bodies are transformed as follows :
	\begin{compactitem}
	\item Each variable in the set returned by $\loopdefs(\Stmt)$ is assigned a non-deterministically chosen value 
		at the start of the loop body and also after the loop body.
	These assignments ensure that values dependent on loop iterations are 
	over-approximated when used inside or outside the loop body. 
       \item The loop iterator $i$ is a special scalar variable. 
	 A loop \Stmt where \fullarrayaccess(\Stmt) holds (case~\ref{transform.stmt.3}) essentially means that loop bound is same as the array size 
	 and array is accessed using loop iterator as index.
	 Hence it is safe to replace array access with \absarr where the 
	 values of loop iterator and index expression match. 
	 To ensure this we equate loop iterator with \absind.
	 This models the behaviour of the original program precisely.
	 However, when $\fullarrayaccess(\Stmt)$ does not hold (case~\ref{transform.stmt.4}), we assign loop iterator $i$ to a 
	 non-deterministically chosen value from the loop bound. 
 \item Each statement in the loop body is transformed 
	as per the transformation rules. 
 \item Finally, the entire loop body is made conditional using a non-deterministically chosen true/false value 
	 when $\fullarrayaccess(\Stmt)$ does not hold. 
	This models the partial
        accesses of array indices which imply that some of the values defined before the loop may reach after the loop.
	However, the transformed loop body is unconditionally executed when $\fullarrayaccess(\Stmt)$ holds.
      \end{compactitem}
 \end{compactitem}

\section{Soundness of Abstraction}
\label{sec:proof}

\externaldocument{transformation}

\begin{figure}[t]
\psset{unit=1mm,arrowsize=1.75}
\begin{tabular}{@{}c@{}}
\begin{pspicture}(0,7)(85,64)
\putnode{m1}{origin}{7}{57}{\psframebox{$i = \Expr_1$}}
\putnode{p}{m1}{2}{-51}{(a) Original Loop}
\putnode{p}{m1}{0}{5}{$l$}
\putnode{n1}{m1}{0}{-13}{\psframebox{$\;\;\;\;\Expr_2\;\;\;\;$}}
\putnode{m2}{n1}{6}{-13}{\psframebox{$\;\;\;\;\Stmt_1\;\;\;\;$}}
\putnode{n2}{m2}{0}{-13}{\psframebox{$\;\;\;\;\Expr_3\;\;\;\;$}}
\putnode{n3}{m2}{-10}{-20}{$m$}
\ncline[nodesepA=1]{->}{n0}{m1}
\ncline{->}{m1}{n1}
\ncline{->}{n1}{m2}
\aput*[3pt](.5){$b$}
\ncline[nodesepB=.15]{->}{m2}{n2}
\aput*[3pt](.5){$a$}
\nccurve[angleA=240,angleB=90,ncurv=.4,nodesepB=1]{->}{n1}{n3}
\ncloop[armA=2.5,angleA=270,angleB=90,offsetB=2,loopsize=-9,linearc=.6]{->}{n2}{n1}
\putnode{n0}{origin}{33}{57}{\psframebox{$\ifkw(nd(0,1))$}}
\putnode{n1}{n0}{12}{-12}{
		\psframebox{$\begin{array}{@{}l@{}}
				u  = nd()
					\\
				i= nd(\loopbound(\Stmt))
								
				\end{array}$}}
\putnode{m2}{n1}{0}{-14}{\psframebox{$\T(\Stmt_1)$}}
\putnode{n3}{m2}{-12}{-11}{\psframebox{$u=nd()$}}
\putnode{n4}{n3}{0}{-8}{$m'$}
\ncline{->}{n0}{n1}
\ncline{->}{n1}{m2}
\aput*[4pt](.45){$b'$}
\ncline[nodesepB=.5]{->}{n3}{n4}
\ncline{->}{m2}{n3}
\aput[0pt](.25){$a'$}
\nccurve[angleA=225,angleB=120,ncurv=.6]{->}{n0}{n3}
\aput[1pt](.8){$d'$}
\putnode{p}{n0}{5}{-51}{(b) $\neg\fullarrayaccess$}
\putnode{p}{n0}{0}{5}{$l'$}
\putnode{n0}{origin}{74}{55}{\psframebox{$\begin{array}{c}
	u = nd()
	\\
	i=\absind			
	\end{array}$}}
\putnode{p}{n0}{0}{-49}{(c) $\fullarrayaccess$}
\putnode{p}{n0}{0}{7}{$l'$}
\putnode{m2}{n0}{0}{-15}{\psframebox{$\T(\Stmt_1)$}}
\putnode{n3}{m2}{0}{-13}{\psframebox{$u=nd()$}}
\putnode{n4}{n3}{0}{-10}{$m'$}
\ncline{->}{n0}{m2}
\aput*[4pt](.45){$b'$}
\ncline[nodesepB=.5]{->}{m2}{n3}
\aput*[4pt](.45){$a'$}
\ncline{->}{n3}{n4}

\end{pspicture}
\end{tabular}

\caption{Transforming a loop}
\label{fig:trans.pgm}
\end{figure}
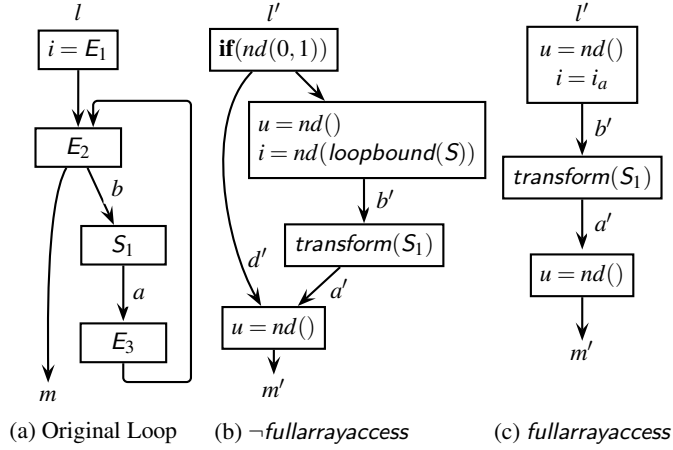


In this section, We argue that the proposed abstraction is sound, i.e. if the transformed
program is safe, then so is the original program. As discussed in Section~\ref{sec:dsa}, the
soundness of abstraction is immediate if the abstract states ``represent'' the
original states. We, therefore, prove that the proposed transformation ensures
that the \emph{represents} relation, $\rightsquigarrow$, holds between abstract and
original states. The proof is by induction on the number of tranformation rules
applied. For the base case, we prove that $\rightsquigarrow$ holds in the beginning -
before applying any transformation (lemma~\ref{lemm:program.start}). 
In the inductive step, we prove
that if $\rightsquigarrow$ holds at some stage during the abstraction, then the subsequent
transformation continues to preserve $\rightsquigarrow$. We prove this by structural
induction on program transformations. Lemma~\ref{lemma:strucinduc.abs} helps us in establishing the base
case of this structural induction.

\begin{lemma}
\label{lemm:program.start}
Let the start of the original program (i.e. the program point just before the code derivable from non-terminal \Stmt in
production \text{$\Prog \rightarrow \Stmt$} in grammar defined in equation(~\ref{eq:grammar.org}) be denoted by $l$. 
The corresponding program point in the 
transformed program \pa, denoted by $l'$, is just after \Init and just before the non-terminal $S$ in production \text{$\Prog \rightarrow \Init \;; \; \Stmt$}
(Grammar in equation~\ref{eq:grammar.trans}).
Let \ms be the state at $l$ and \smsa be the set of states at $l'$ in \pc and \pa respectively.
Then,
\text{$\smsalp \rightsquigarrow\msl$}.
\end{lemma}
\begin{proof}
	Let \text{$\init(x)$} denote the initial value in variable \text{$x \in \svar - \{\arrvar\}$}.
Similarly, let
\text{$\init(a)$} denote the initial value common to all indices in array \arrvar.
Then, \msl in \pc is: 
\begin{align*}
\msl  = &
	\left\{
	\left(x,\init(x)\right) \mid 
	\left(x \in \svar - \{\arrvar\}\right) 
	\right\} \; \cup \;
	\\
	&
	\;\;\,
	\left\{
		\left(\loc(a[i]),\init(a)\right) \mid 0 \leq i \leq \size(\arrvar)
	\right\}
\end{align*}

Since $l'$ appears after \Init
in \pa, the witness index variable is assigned a non-deterministic
value corresponding to some
index of the array  (case~\ref{transform.prog.1} in Figure~\ref{fig:transform}).
All other variables (including the witness variable representing array) retain their initial values.  Then there exist $\msalp \in \smsalp$ such that,
\begin{align*}
	 \msalp & = 
	 X 
	\cup
	\left\{
		\left(\absind,i\right)
	\right\} 
\intertext{where, $X$ contains the variables with their initial values.}
X & = 
	\left\{
	\left(x,\init(x)\right) \mid 
	\left(x \in \svar - \{\arrvar\}\right) 
	\right\} \cup
	\left\{
	\left(\absarr,\init(\arrvar)\right)
	\right\} 
	\\
\end{align*}
It can be immediately seen that {\smsalp $\rightsquigarrow$ \msl} (as per Definition~\ref{def:abselement1} and ~\ref{def:abselement2}) 
for every value of \absind. 
Note that \absarr and \absind are the witness variable and witness index respectively.  


\end{proof}

\begin{lemma}
\label{lemm:rhs1}
Let $\msl$ be a state at a program point $l$ in \pc 
and $\smsalp$ be set of states at the corresponding program point $l'$ in transformed program \pa.
Consider an arbitrary expression \text{$e \in \sexpr$} just after $l$ in original program \pc.
Then, 
\[
	\text{$\smsalp \rightsquigarrow\msl \Rightarrow \eval{\T(e)}{\smsalp} \supseteq \{\eval{e}{\msl}\}$}.
\]
\end{lemma}

\begin{proof}
Since $e$ is derived from \Expr (grammar~\ref{eq:grammar.org}),
we prove the lemma by structural induction on the productions for \Expr
by substituting \Lvalue by its right hand side.
\[
\begin{array}{rl}
\Expr \rightarrow &
	\Expr_1 \oplus \Expr_2
	\;\; \big\lvert \;\; a[\Expr_1] 
	\;\; \big\lvert \;\; \var
	\;\; \big\lvert \;\; c
\end{array}
\]
For the base case, that is when \Expr derives \var or $c$, \text{$\T(\Expr) = \Expr$} 
(case \ref{transform.expr.3} in Figure~\ref{fig:transform}) and the lemma holds trivially.

Assume it holds for an arbitrary \Expr then




\begin{itemize}
	\item When \Expr derives $a[\Expr_1]$, from the base case it implies that \text{$\eval{\T(\Expr_1)}{\smsalp} \supseteq \eval{\Expr_1}{\msl}$}.
	This means that in \pa more indices are read than in the \pc. Hence the lemma holds.	
		
%

	\item When \Expr derives \text{$\Expr_1 \oplus \Expr_2$} - From the hypothesis 
\text{$\smsalp \rightsquigarrow\msl$} it follows that \\ \text{$\eval{\T(\Expr_1)}{\smsalp} \supseteq \eval{\Expr_1}{\msl}$} and
\\ \text{$\eval{\T(\Expr_2)}{\smsalp} \supseteq \eval{\Expr_2}{\msl}$}. 

\end{itemize}

Therefore \text{$\eval{\T(\Expr)}{\smsalp} \supseteq \eval{\Expr}{\msl}$}.

\end{proof}

\begin{lemma} 
Let $l$ and $m$ be the program points just before and after a statement $s$ in \pc and let $\msl$ and $\msm$ be 
the states at $l$ and $m$ respectively.
Let $l'$ and $m'$ be the program points just before and after the corresponding transformed statement $\T(s)$ in \pa.
Let \smsalp and \smsamp be the set of states at $l'$ and $m'$ respectively. 
Then, \text{$\smsalp \rightsquigarrow \msl \Rightarrow \smsamp \rightsquigarrow \msm$}.


\label{lemma:strucinduc.abs}
\end{lemma}
\begin{proof}
	Statement $s$ is derived from non-terminal \Stmt in the grammar~\ref{eq:grammar.org}. 
	Hence, we prove the lemma by structural induction on \Stmt.

For the base case, consider the non-recursive productions of \Stmt:
\begin{enumerate}
	\item   For the right hand side \text{$(\Lvalue=\Expr)$}, there are two cases to
      consider-
 \begin{itemize}
	 \item When \text{$(\Lvalue=\Expr)$} derives \text{$\arrvar[\expr_1]  = \expr_2$} 
      (case \ref{transform.stmt.1} in Figure~\ref{fig:transform}), 
		then 
		since $\T(\expr_2)$ always over-approximates 
			(Lemma~\ref{lemm:rhs1}) then $\smsamp \rightsquigarrow\msm$ trivially.

      \item When \Lvalue is not an array, the $\transform(\Expr)$ is assigned to unchanged \Lvalue.
	      Since $\transform(\Expr)$ is always over-approximated $\smsamp \rightsquigarrow\msm$. 
      \end{itemize}
\item For the right hand side \text{$\assertkw(\Expr)$},
      the states do not change and the lemma holds trivially.
\end{enumerate}
For the inductive step, we consider the recursive productions of \Stmt.
\begin{enumerate}
\item For the right hand side \text{$\Stmt_1;\Stmt_2$},
	the lemma holds by induction. 

\item For the right hand side \text{$\ifkw(\Expr) \; \Stmt_1 \; \elsekw \; \Stmt_2$},
	the set of values of \text{$\T(\Expr)$} is over-approximated (lemma~\ref{lemm:rhs1}) and it does not alter the state.
	Hence, $\rightsquigarrow$ holds at the start of $\Stmt_1$ and $\Stmt_2$.
       By induction, the lemma holds whether $\Stmt_1$ is executed or  $\Stmt_2$ and hence it holds after the \emph{if} statement. 

\item For the right hand side \text{$\ifkw(\Expr) \; \Stmt_1 $}, the argument is similar to the
      case of \text{$\ifkw(\Expr) \; \Stmt_1 \; \elsekw \; \Stmt_2$}.

\item For the right hand side \text{$\forkw(i = \Expr_1;\;\Expr_2;\;\Expr_3) \; \Stmt_1$},
	let $b$ and $a$ be the program point just before and after $\Stmt_1$ 
	and $b'$ and $a'$ be the program point just before and after $\T(\Stmt_1)$.   
	We transform the loop in two ways as illustrated in Figure~\ref{fig:trans.pgm}.
	The loop body has either an assignment statement of the form, \text{$\arrvar[\expr_1]  = \expr_2$} or an expression of the form $\arrvar[\expr]$ or both.
	\begin{enumerate}
	\item We first show that the lemma holds when \fullarrayaccess does not hold (case~\ref{transform.stmt.4})
		In this case $a[c]$ may not be referred to or updated, where $(\absind, c) \in \msalp$.

	\begin{itemize} 
		\item First lets consider the path $l'$ to $d'$ to $m'$ (see Figure~\ref{fig:trans.pgm}.b). In the transformed program since any path can be chosen non-deterministically, this is the path we take when $a[c]$ is not referred to or updated.

		Since in this path we add non-deterministic assignments to variables in the set \loopdefs(\Stmt),
		$\rightsquigarrow$ holds at the end of the \forkw statement.
		
		\item For the other path (see Figure~\ref{fig:trans.pgm}.b), which goes through the transformed loop body,
			we establish the lemma in following steps-
			\begin{enumerate}
			\item \text{$\msalp \rightsquigarrow \msl \Rightarrow \msabp \rightsquigarrow \msb$}
			\item \text{$\msabp \rightsquigarrow \msb \Rightarrow \msaap \rightsquigarrow \msA$}
			\item \text{$\msaap \rightsquigarrow \msA \Rightarrow \msamp \rightsquigarrow \msm$}	
			\end{enumerate}
			Step-(ii) follows from induction.
			Step-(iii) follows a similar argument as in the case above, since the variables that get modified obtain a \\non-deterministic value. 
			Whenever $\arrvar \in \loopdefs(\Stmt)$ or array is not modified by $\Stmt$, a similar argument can be used for step - (i) too.  
			
			Now for the case when $a[c]$ is modified by $\Stmt$ and $\arrvar \notin
\loopdefs(\Stmt)$ consider $\ms$ at the start of the iteration in which $a[c]$
is modified. At this point $\msabp \rightsquigarrow \msb$ because all variables have been
assigned non-deterministic values, $\msalp \rightsquigarrow \msl$ and $a[c]$ has not been
modified so far. Therefore it inductively follows $\msaap \rightsquigarrow \msA$ and therefore $\msamp \rightsquigarrow \msm$.
Since, $\msamp \in \smsamp$ hence $\smsamp \rightsquigarrow \msm$.
\end{itemize}

	\item In the case of \fullarrayaccess, $a[c]$ is always modified by $\Stmt$. We can therefore use the same argument
		as in the case when $a[c]$ is modified, above. 

	\end{enumerate}

\end{enumerate}
Hence the lemma.
\end{proof}

\begin{theorem}
If the assertion \assert is violated in the original program \pc, then it will be violated in transformed program \pa also. 
\end{theorem}
\begin{proof}
Let the assert get violated for some $a[c]$. Since $\absind$ is initialized
non-deterministically it can take the value $c$ and we have shown in Lemma~\ref{lemm:rhs1} that
all expressions in $\pa$ are over-approximated. Lemma~\ref{lemm:program.start} and Lemma~\ref{lemma:strucinduc.abs} ensure the premise for Lemma~\ref{lemm:rhs1}.
Hence the theorem follows.

\end{proof}

\newcommand{\Defset}{\text{\sf\em$\mathbb{S}_\text{\em def}$}\xspace}
\newcommand{\ImpVarset}{\text{\sf\em$\mathbb{V}_\text{\em imp}$}\xspace}
\newcommand{\ImpExpset}{\text{\sf\em$\mathbb{E}_\text{\em imp}$}\xspace}

\section{Precision}
\label{sec::prec}

We characterize the assertions for which our transformation is precise -- an assertion will fail in 
\pa if and only if it does so in \pc.
We denote such an assertion as \assertinv. We focus on \assertinv in a loop.
In case if array accesses outside loops, we do not claim precision. Our experience shows that such situations
are rare in programs with large arrays.

Our transformations
replace array access expressions and loop statements while the statements involving
scalars alone outside the loop remain unmodified. Hence precision criteria need to focus on the statements within loops and not outside it.

Let assertion \assertinv appear in a loop statement $S\!_\assertinv$.
Let \ImpVarset be the set of variables and \ImpExpset be the set of array access expressions on which \assertinv is 
data or control dependent within the loop $S\!_\assertinv$.
Let the set of loop statements from where definitions reach \assertinv be denoted by \Defset, note that this
set is a transitive closure for data dependence.
Our technique is precise when:
\begin{itemize}
\item $\fullarrayaccess(S)$ holds for each \text{$S \in \{S\!_\assertinv\} \cup \Defset$} \hfill(rule $l_1$)
\item If  $a[e] \in \ImpExpset$ then
	\begin{itemize}
	\item the index expression $e = i$ where $i$ is the loop iterator of loop $S\!_\assertinv$  \hfill 
		(rule $a_2$)
 	\item $\arrvar \notin \loopdefs(S)$ where \text{$S \in \{S\!_\assertinv\} \cup \Defset$} \hfill (rule $a_3$)
\end{itemize}
\item  If $\var \in \ImpVarset$ then
	$\var \notin \loopdefs(S)$ where \text{$S \in \{S\!_\assertinv\} \cup \Defset$} \hfill (rule $s_4$)
\item For an assignment statement of the form $a[e_1] = e_2$ in loop $S$ where $S \in \Defset$,  
      \begin{itemize}
       \item if $e_2$ is an array access expression then it must be of the form $a[i]$ where $i$ is the loop iterator 
	       of loop $S$ 	\hfill (rule $d_5$)
	\item if $e_2$ is \var then $\var \notin \loopdefs(S)$ where \text{$S \in \Defset$} \hfill 
		(rule $d_6$)
	\end{itemize}

\end{itemize}

\begin{theorem}
\label{th:comp}
If an assertion \assertinv that satisfies above rules 
holds in the original program \pc, then it will hold in the 
transformed program \pa also.

\end{theorem}

\begin{proof}
	The transformed program is over-approximative because our transformation rules (~\ref{transform.stmt.3},~\ref{transform.stmt.4},~\ref{transform.expr.2}) 
	introduce non-deterministically chosen values.
	We prove this theorem by showing that if assertion is of the form \assertinv then 
	none of these transformation rules introduce non-deterministically chosen values in the transformed program.

\begin{itemize}
	\item Since rule $l_1$ holds unconditionally, loops will be transformed as per~\ref{transform.stmt.3} and not~\ref{transform.stmt.4} and
	extra paths are not added. Also, the assignment $i=\absind$ at the start of the loop
	ensures that  the condition is true for each transformed array access expression. 
	\item When rule $a_2$ holds, rule $l_1$ also holds and $a[e]$ is always 
	replaced by \absarr (case~\ref{transform.expr.2}).
	\item When rule $a_3$ holds, assignment $\absarr=nd()$ is not added (case~\ref{transform.stmt.3}).
	\item When rule $s_4$ holds, assignment $\var=nd()$ is not added (case~\ref{transform.stmt.3}).
	\item When rule $d_5$ holds, rule $l_1$ holds and $a[e]$ in RHS is
		 replaced by \absarr (case~\ref{transform.expr.2}).
	\item When rule $d_6$ holds, scalars used in RHS are not assigned a non-deterministically chosen values.
	
\end{itemize}


\end{proof}

Rules $s_4$ and $d_6$ are strong requirements for ensuring that non-deterministically chosen values do not reach \assertinv.
We can relax these rules 
when $\var \in \loopdefs(S\!_\assertinv)$ when: 
\begin{itemize}
\item A definition of \var appears before the use of \var in RHS in the loop. 
\item Variable \var is defined with a constant value or using loop iterator $i$ only.
\end{itemize}

Since the original assignments to \var are retained in the transformed loop body, assignment of \var to non-deterministically chosen value ($\var = nd()$) 
gets re-defined.  
Also, when \var is defined with a constant or $i$ ($i=\absind$ is added for $S\!_\assertinv$),
its value is not over-approximated.

The assertion in Figure~\ref{Mex1} is \assertinv since it satisfies all the rules:
\begin{itemize}
	\item $S\!_\assertinv$ is the loop at line 14 containing the assertion. 
	$\Defset$ contains the loop at line 8. For both these loops, rule $l_1$ holds.
	\item \ImpExpset consists of the three array access expressions at line 16. 
	Rule $a_2$ and $a_3$ hold for all the three expressions. 
	\item \ImpVarset consist of the loop iterator \emph{i} hence rule $s_4$ holds.
	\item For assignments at line 11 and 12 of the lone loop in $\Defset$,
		\begin{compactitem}
		\item rule $d_5$ holds.
		\item \var is \emph{k} which is in \loopdefs(\Stmt). But we are using relaxed form of rule $d_6$ here. \emph{k} is defined 
		using loop iterator \emph{i} at line 10. Hence $d_6$ hold.
		\end{compactitem}
\end{itemize}

Since the assertion is \assertinv the transformed program in Figure~\ref{TEx1} is not over-approximated and hence BMC is able to prove the assertion.

\section{Optimization}
\label{opti}

We have optimized the transformations such that we model the behaviour of original program 
as precisely as possible without impacting the soundness of our approach.
\begin{itemize}
\item There can be several ways in which loops do not iterate over all the indices of an array. Depending on the way loop is written we can
	improve the transfomations instead of applying rule~\ref{transform.stmt.4}.
For instance, when a loop has \breakkw or \continuekw statement, we keep the loop header such that it iterates only once. 
Another instance is when the increment expression of a loop is incremented monotonically but not by 1.
We remove the loop header and the transformed loop body is controlled by a conditional statement, say $\emph{\ifkw(i\%3)}$ for increment expression $i+3$.
	In both the cases, we still retain additional assignments for variables in set returned by \loopdefs(\Stmt) and loop iterator.
\item We further gain precision by using range analysis~\cite{kumar2013precise}. This analysis returns the range of value (statically) 
	that a variable can take at a program point.
	This range is passed as a parameter to the function returning non-deterministic value. 
	Thus restricting the choice from the given range.
\end{itemize}

%

\section{Implementation}
\label{imple}
We have implemented our transformation engine using a static
analysis framework called PRISM developed at TRDDC, Pune~\cite{chimdyalwar2011effective,khare2011static}. 
Our implementation supports
\textsc{Ansi}-C programs with 1-dimensional arrays.
Structure are field sensitive: x.a and x.b are treated as separate entities.
Pointers are handled using flow sensitive and context insensitive points-to analysis
as implemented in PRISM.
Using static analysis we populate functions such as \fullarrayaccess(\Stmt) and \loopdefs(\Stmt).
We use information from these functions to transform the program. Range analysis provides information for \emph{assume} statements.

Our transformation engine takes the original C program as input and generates a valid C program without array and loop.
It adds proper declarations for witness variable and witness index and assume statements too.
\subsection{Handling of Multiple Arrays}
Though, we restricted our formalization to single array in a program, our tool can handle multiple arrays in a program with following enhancements:
\begin{itemize} 
\item Function $\fullarrayaccess(\Stmt)$ returns true when all the arrays accessed in the loop \Stmt are of \emph{same} sizes. This ensures that
all indices of each array is accessed. 
When multiple arrays with \emph{different} sizes are being read or written $\fullarrayaccess(\Stmt)$ conservatively returns false.
\item We also add an assignment at the start of the program where witness index of each array of \emph{same} size are equated
to receive same non-deterministically chosen value. 
This ensures that for a loop, where $\fullarrayaccess(\Stmt)$ returns true, same index of each array is considered in a run of transformed program.
\end{itemize}
\begin{figure*}[t]
\scriptsize
\begin{minipage}{0.45\textwidth}
\begin{lstlisting}[frame=single,language=c,escapeinside={(*@}{@*)},basicstyle=\ttfamily]
01. int a[100000],b[100000],
        c[100000];
02. int i,v;

03. main()
04. {
/* Loop1 */
05.  for(i = 0; i < SIZE; i++)
06.  {
07.   v = input();
08.   a[i]= v;
09.   v = input();
10.   b[i] = v;
11.  }
/* Loop2 */
12.  for(i = 0; i < 100000; i++)
13.  {
14.   c[i] = a[i] + b[i];
15.  }
/* Loop3 */
16.  for(i = 0; i < 100000; i++)
17.  {
18.   (*@ \color[rgb]{0.75,0.164,0.164}{assert}@*)(c[i] == a[i] + b[i]);
19.  }
20. }
\end{lstlisting}
\caption{Multiple Array - Program 1}
\label{Marr1}
\end{minipage}
\quad
\quad
\begin{minipage}{0.50\textwidth}
\begin{lstlisting}[frame=single,language=c,escapeinside={(*@}{@*)},basicstyle=\ttfamily]
01. int x_a,x_b,x_c;
02. int i_a,i_b,i_c;   
03. int i,v;

04. main()
05. {
06.   i_a=i_b=i_c=nd(0,100000);
/* Loop1 body */
07.   i=i_a; i=i_b;
08.   v = nd ();
09.   v = input();
10.   (i==i_a)? x_a= v : v;
11.   v = input();
12.   (i==i_b)? x_b= v : v;
13.   v = nd();
/* Loop2 body */
14.   i=i_a; i=i_b; i=i_c;
15.   (i==i_c)?x_c=((i==i_a)?x_a:nd()
      +(i==i_b)?x_b:nd()):
      ((i==i_a)?x_a:nd()
      +(i==i_b)?x_b:nd());
/* Loop3 body */
16.   i=i_a; i=i_b; i=i_c;
17.   (*@ \color[rgb]{0.75,0.164,0.164}{assert}@*)(i==i_c)?x_c
      ==((i==i_a)?x_a:nd()
      +(i==i_b)?x_b:nd())
      :((i==i_a)?x_a:nd()
      +(i==i_b)?x_b:nd()); 
18. }
\end{lstlisting} 
\caption{Transformed Program}
\label{TMarr1}
\end{minipage}
\end{figure*}

Consider the example in Figure~\ref{Marr1} with multiple arrays. 
In the first loop (lines 5-11), elements of array a and b are assigned with 
value provided by the user.
The second loop(line 12-15) assigns sum of each element of array a and array b to the corresponding element of array c. 
Third loop asserts that each index of array c contains the sum of elements at the same 
index of array a and b. This program is safe.

In the transformed program, shown in Figure~\ref{TMarr1}, 
witness index $i_a$, $i_b$ and $i_c$ are assigned with same non-deterministically chosen value at 
line 6 since array a, b and c are of same size.  
For all three loops \fullarrayaccess(\Stmt) hold,
hence loop iterator is assigned with witness index of the arrays accessed inside the loop 
(line 7, 14 and 16).
Other transformations are as explained in earlier sections.
BMC successfully verified the transformed program. 

\begin{figure*}[t]
\scriptsize
\begin{minipage}{0.45\textwidth}
\begin{lstlisting}[frame=single,language=c,escapeinside={(*@}{@*)},basicstyle=\ttfamily]
01. int i,x = 0,y = 0;
02. int a[100000],b[50000]; 

03. main()
04. { 
05.   x = user_input();
06.   y = user_input();
/* Loop1 */
07.   for(i = 0 ; i< 50000; i++)
08.   { 
09.     a[i] = x; 
10.     b[i] = y; 
11.     a[i+50000] = x*2;
12.     x++; y++;
13.   }
/* Loop2 */
14.   for(i=0; i < 50000; i++)
15.   {
16.     (*@ \color[rgb]{0.75,0.164,0.164}{assert}@*)(a[i]*2==a[i+50000]);
17.   }
18. }
\end{lstlisting}
\caption{Multiple Array - Program 2}
\label{Marr2}
\end{minipage}
\quad
\quad
\begin{minipage}{0.50\textwidth}
\begin{lstlisting}[frame=single,language=c,escapeinside={(*@}{@*)},basicstyle=\ttfamily]
01. int i_a, i_b;  
02. int i, x_a, x_b, x =0, y=0;

03. main()
04. {
05.   i_a = nd(0,99999);
06.   i_b = nd(0,49999);  
07.   x = user_input();
08.   y = user_input();
/* Loop1 body */
09.   if(nd(0,1))
10.   {	
11.     x = nd(); y = nd();
12.     i = nd(0,49999);
13.     (i==i_a)?x_a = x:x; 	
14.     (i==i_b)?x_b = y:y;
15.     (i+50000==i_a)?x_a=x*2:x*2;
16.     x++; y++;
17.     } 
18.    x = nd(); y = nd();
/* Loop2 body */
19.   i = nd(0,49999);
20.   (*@ \color[rgb]{0.75,0.164,0.164}{assert}@*)((i==i_x)?(x_a*2):nd()==
		(i+50000==i_x)?(x_a):nd());
21. }
\end{lstlisting} 
\caption{Transformed Program} 
\label{TMarr2}
\end{minipage}
\end{figure*}

The program in Figure~\ref{Marr2} has multiple arrays too but with different sizes.
In the transformed program as shown in Figure~\ref{TMarr2} 
witness index $i_a$ and $i_b$ are initialized at line 5 and 6 
to a non-deterministically chosen value 
from their respective range of the indices 
In \emph{Loop1} \fullarrayaccess(\Stmt) does not hold 
since different sized arrays are accessed.
Since the transformed program is over-approximated,  
BMC could not prove this assertion.
However, our approach for multiple arrays is sound in all scenarios. 

%

%


\section{Experimental Evaluation}
\label{exp}
The experiments are
performed on a 64-bit Linux machine with 16 Intel Xeon processors running at
2.4GHz, and 20GB of RAM. 
Although we could take
any off-the-shelf BMC for C program to verify the transformed code, we use CBMC
in our experiments as it is known to handle all the constructs of
\textsc{Ansi}-C.
We performed experiments under three setups- 
\begin{inparaenum}[\em (i)]
\item SV-COMP benchmarks
\item 3 large real-life applications
\item a set of small programs based on the patterns 
observed in these applications.	
\end{inparaenum}

\subsection{Setup 1 : SV-COMP Benchmarks}

\begin{table}[t]
\centering
\caption{Results on SV-COMP Benchmark Programs.}
\label{svcomp_table}
\begin{tabular}{|c|c|c|c|c|c|}
\hline
\#programs = 118 & \begin{tabular}[c]{@{}c@{}}\#correct\\ true\end{tabular} & \begin{tabular}[c]{@{}c@{}}\#correct \\ false\end{tabular} & \begin{tabular}[c]{@{}c@{}}\#incorrect \\ true\end{tabular} & \begin{tabular}[c]{@{}c@{}}\#incorrect \\ false\end{tabular} & \#no result \\ \hline \hline
Expected Results & 84 & 34 & - & - & 0 \\ \hline \hline
$CBMC_\alpha$ & 47 & 06 & 06 & 0 & 59 \\ \hline
$CBMC_\beta$ & 9 & 5 & 0 & 0 & 104 \\ \hline
Transformation+$CBMC_\beta$ & 25 & 34 & 0 & 59 & 0 \\ \hline \hline
\multicolumn{6}{|c|}{$CBMC_\alpha$ - SV-COMP2016 (unsound) CBMC, $CBMC_\beta$ - sound CBMC 5.4} \\ \hline
\end{tabular}
\end{table}

SV-COMP
benchmarks~\cite{svcompbench} contain an established set of programs under
various categories intended for comparing software verifiers.
Results for \emph{ArraysReach}\footnote{Programs in \emph{ArrayMemSafety}
access arrays without using index and cannot be transformed.} from the
\emph{array} category
for CBMC used in SV-COMP 2016 ($CBMC_\alpha$), CBMC 5.4 ($CBMC_\beta$) and CBMC
5.4 on transformed programs (Transformation+$CBMC_\beta$) are
consolidated\footnote{Case by case results available at \\
\htmladdnormallink{https://sites.google.com/site/datastructureabstraction/home/sv-comp-benchmark-evaluation-1}{https://sites.google.com/site/datastructureabstraction/home/sv-comp-benchmark-evaluation-1}}
in Table~\ref{svcomp_table}.
\emph{ArraysReach} has 118 programs. $CBMC_\alpha$, an unsound version of CBMC,
gave correct results for 53 programs. 
However, $CBMC_\beta$ gave correct results for 14 programs.
We compare the results of Transformation+$CBMC_\beta$ on three criteria:

\begin{compactitem}
\item Scalability: it scaled up for all 118 programs.
\item Soundness: it gave sound results for all 118 programs. For the 6
programs for which $CBMC_\alpha$ gave unsound results, our results are not only
sound but are also precise.
\item Precision: it gave precise results for 59 programs. Out of these
$CBMC_\alpha$ ran out of memory for 45 programs ($CBMC_\alpha$ ran out of
memory for 14 additional programs).  On the other hand, 22 true
programs reported correctly by $CBMC_\alpha$ were verified as false by
Transformation+$CBMC_\beta$.  Transformation+$CBMC_\beta$ verified 25 programs
as true which did not include 8 of the programs reported correctly as true by
$CBMC_\beta$.	
\end{compactitem}

Our technique is imprecise for the other 59 of 118 programs as they do not comply with
the characterization of precision provided in Section~\ref{sec::prec}.  As can
be seen, there is a trade-off between scalability and precision.  From the view
point of reliability of results, soundness is the most desirable property of a
verifier. Our technique satisfies this requirement. Further, it not only scales up  but is also precise 
in many situations implying its practical usefulness.

\subsection{Setup 2 : Real-life Applications}

\begin{table}[t]
\centering
\caption{Real-life Application Evaluation}
\label{RealApp_table}
\begin{tabular}{ccccccccccc}
\hline
\multicolumn{4}{|c|}{Application details} & \multicolumn{3}{c|}{Sliced+CBMC} & \multicolumn{3}{c|}{\begin{tabular}[c]{@{}c@{}}Sliced\\ +Transformation \\ +CBMC\end{tabular}} & \multicolumn{1}{c|}{\multirow{2}{*}{\begin{tabular}[c]{@{}c@{}}\%\\ False\\  Positive \\ Reduction\end{tabular}}} \\ \cline{1-10}
\multicolumn{1}{|c|}{Name} & \multicolumn{1}{c|}{\begin{tabular}[c]{@{}c@{}}Size\\ (LoC)\end{tabular}} & \multicolumn{1}{c|}{ $\% loop^{full}$} & \multicolumn{1}{c|}{\#Asserts} & \multicolumn{1}{c|}{\#P} & \multicolumn{1}{c|}{\#F} & \multicolumn{1}{c|}{\#T} & \multicolumn{1}{c|}{\#P} & \multicolumn{1}{c|}{\#F} & \multicolumn{1}{c|}{\#T} & \multicolumn{1}{c|}{} \\ \hline \hline
\multicolumn{1}{|c|}{\it{navi1}} & \multicolumn{1}{c|}{1.54M} & \multicolumn{1}{c|}{100} & \multicolumn{1}{c|}{63} & \multicolumn{1}{c|}{0} & \multicolumn{1}{c|}{0} & \multicolumn{1}{c|}{63} & \multicolumn{1}{c|}{52} & \multicolumn{1}{c|}{1} & \multicolumn{1}{c|}{10} & \multicolumn{1}{c|}{82.5} \\ \hline
\multicolumn{1}{|c|}{\it{navi2}} & \multicolumn{1}{c|}{3.3M} & \multicolumn{1}{c|}{93.4} & \multicolumn{1}{c|}{103} & \multicolumn{1}{c|}{0} & \multicolumn{1}{c|}{0} & \multicolumn{1}{c|}{103} & \multicolumn{1}{c|}{95} & \multicolumn{1}{c|}{1} & \multicolumn{1}{c|}{7} & \multicolumn{1}{c|}{92.2} \\ \hline
\multicolumn{1}{|c|}{icecast\_2.3.1} & \multicolumn{1}{c|}{336K} & \multicolumn{1}{c|}{59.1} & \multicolumn{1}{c|}{114} & \multicolumn{1}{c|}{0} & \multicolumn{1}{c|}{0} & \multicolumn{1}{c|}{114} & \multicolumn{1}{c|}{53} & \multicolumn{1}{c|}{61} & \multicolumn{1}{c|}{0} & \multicolumn{1}{c|}{46.5} \\ \hline \hline
\multicolumn{11}{|c|}{\begin{tabular}[c]{@{}c@{}} $loop^{full}$ - loop \Stmt where \fullarrayaccess(\Stmt) holds, \\ P - Assertion Proved, F - Assertion Failed, T - Timeout\end{tabular}} \\ \hline 

\end{tabular}
\end{table}

We applied our technique on 3 real-life applications - 
{\it navi1} and {\it navi2} are industry codes implementing the navigation
system of an automobile and icecast\_2.3.1 is an open source project for
streaming media~\cite{icecast}. To verify a meaningful property, we used \emph{Null
Pointer Dereference} (NPD) warnings generated by a sound static analysis tool build using PRISM\footnote{TCS
Embedded Code Analyzer (TCS ECA)\\
\htmladdnormallink{http://www.tcs.com/offerings/engineering\_services/Pages/TCS-Embedded-Code-Analyzer.aspx}{http://www.tcs.com/offerings/engineering\_services/Pages/TCS-Embedded-Code-Analyzer.aspx}}. 
PRISM performs weak updates for arrays (similar to array smashing) and hence generates a large number of warnings on arrays, most of which are false.
We appended assertions only for an array or an array member dereferencing as follows. 
Lets say the dereference expression is $*a[i].p$.  A
statement $assert(a[i].p!=null)$ is added in the code just before the statement
containing dereference expression. We then slice{\footnote{PRISM implements~\cite{Horwitz1990Inter} for slicing.}} these programs as per assertion.


Table~\ref{RealApp_table} shows the consolidated results of
our experiments. CBMC did not
scale on the sliced programs. However, after transformation (sliced+transformation+CBMC) CBMC proved 200 
out of 280 assertions, taking 12 minutes on average for transformation+verification. This
is much less in comparison to the time given to CBMC for sliced programs
(sliced+CBMC), which was 30 minutes. 

To verify the correctness of our implementation, we analyzed the warnings
manually. We found that all 280 warnings were false, implying that all the
assertions should have been proved successfully.

\begin{compactitem}
\item Scalability: CBMC could scale up for such large applications because there are no loops in transformed programs.
However, CBMC could not scale for 17 cases even after transformation because of long recursive call chain 
through function pointers. 
\item Precision: Our technique proved 200 assertions, where
all the conditions for precision mentioned in Section~\ref{sec::prec} get fulfilled. In all these cases since
$\fullarrayaccess(\Stmt)$ hold, the witness variable gets precise value in each run of BMC for each value of 
witness index. 
63 of the assertions could not be proved since array definitions reaching at the assertion were from the
loops where $\fullarrayaccess(\Stmt)$ did not hold and values received by witness variable is over-approximative. 
\end{compactitem}

Replacing array expressions with witness variable enables elimination of loops. 
However, by using witness index BMC simulates loop iterations with providing 
run-based value to witness variable.
Note that the number of false warnings
eliminated in an application is proportional to the number of loops for which
\fullarrayaccess(\Stmt) hold. Over a diverse set of applications, we found that our technique
could eliminate 40-90\% of false warnings. This is a significant value addition
to static analysis tools 
that try to find defects and end up generating a large
number of warnings. 

\subsection{Setup 3 : Our Benchmarks}

\begin{table}[]
\centering
\caption{Results of various BMCs on Benchmark Programs.}
\label{suite1_table}
\begin{tabular}{|c|c|c|c|c|c|c|c|c|c|c|c|}
\hline
\multirow{2}{*}{Name} & \multirow{2}{*}{\begin{tabular}[c]{@{}c@{}}Expected\\ Result\end{tabular}} & \multicolumn{2}{c|}{CBMC} & \multicolumn{2}{c|}{LLBMC} & \multicolumn{2}{c|}{\begin{tabular}[c]{@{}c@{}}SMACK+\\ CORRAL\end{tabular}} & \multicolumn{2}{c|}{ESBMC} & \multicolumn{2}{c|}{CPAChecker} \\ \cline{3-12} 
 &  & OP & TP & OP & TP & OP & TP & OP & TP & OP & TP \\ \hline
motivatingExample & S & O & S & T & S & T & S & O & S & T & S \\ \hline
Relation\_Mem & S & T & S & T & S & T & S & T & S & T & S \\ \hline
break\_safe & S & O & S & T & S & T & S & O & S & T & S \\ \hline
Nested\_Struct & S & T & S & T & S & T & S & O & S & T & S \\ \hline
partial\_lesser\_bound & S & T & U & T & U & O & U & T & U & T & U \\ \hline
\multicolumn{12}{|c|}{\begin{tabular}[c]{@{}c@{}}OP-Original Program, TP-Transformed Program \\ S-Safe U-Unsafe T-Time out O-Out of memory\end{tabular}} \\ \hline
\end{tabular}
\end{table}

      

Since the industry code cannot be shared, we created a set of programs based on the patterns observed in real-life applications.
These benchmark programs are provided at 
\htmladdnormallink{https://sites.google.com/site/datastructureabstraction/}{https://sites.google.com/site/datastructureabstraction/}.

We ran following model checkers that implement different techniques:
\begin{inparaenum}[\itshape(i)\upshape]
\item CBMC 5.4, a bounded model checker for C, 
\item LLBMC 2013.1, a low-level bounded model checker~\cite{llbmc}, and
\item Smack + Corral\footnote{SV-COMP 2015 winner in \emph{array} category} \cite{haran2015smack},
\item ESBMC, a context-bounded model checker\footnote {SV-COMP 2016 winner in \emph{array} category} \cite{esbmc}, and   
\item CPA Checker, a tool based on lazy abstraction and interpolation~\cite{cpachecker}.
\end{inparaenum}
Time out was set for 900 seconds. 

All 5 programs are safe. Table~\ref{suite1_table} summarizes the result of various model checkers on the original and the transformed programs.
None of the model checkers scaled on the original programs, while 
they scaled-up for all the transformed programs. Out of 5, 4 transformed programs are reported as safe while 1 is reported as unsafe.
The 5 programs reported as safe comply to the precision criteria mentioned in Section~\ref{sec::prec}, specially that the 
\fullarrayaccess(\Stmt) hold for all the loops. While the 1 program is reported as unsafe due 
to over-approximations introduced since \fullarrayaccess(\Stmt) did not hold.

\section{Related Work}
\label{relwork}

There is a sizeable body of work on reasoning about values of array elements.
We give a brief overview of the relevant literature here.

Blanchet et al. proposed the use of abstract elements for array
indices~\cite{blanchet2002design}. 
Though of limited practicability, they
suggested two approaches - one that uses an abstract element for each index in
the array i.e. {\it array expansion}, and another where a single abstract
element is used for the entire array i.e. {\it array smashing}. Array expansion
is precise but it cannot be used effectively for large arrays. Array
smashing, on the other hand, allows handling arbitrary arrays efficiently, but
suffers significant precision losses due to its inability to perform strong
updates. As a result, it cannot be applied to prove the correctness of our
example in Figure~\ref{Mex1}.

While expansion and smashing are at two extremes in their use of abstract
elements, there are several ``midway'' approaches that are aimed at combining
the benefits of both. Gopan et al.~\cite{gopan2005framework} proposed that
those elements that are read or written could be dynamically expanded to
incorporate strong updates, while the remaining elements could be smashed. This
was extended to partitioning arrays into symbolic intervals spanning the entire
array, and using an abstract variable to represent each such
interval~\cite{halbwachs2008discovering}. To avoid the exponential
multiplication of array slices that prevented these techniques from scaling,
Cousot et al. proposed an improvement~\cite{cousot2011parametric} that
automatically, semantically divides arrays into consecutive non-overlapping
segments. {\tt Clousot}~\cite{clousotTool}, a tool that implements this
improvement, indeed scales better than
\cite{gopan2005framework,halbwachs2008discovering} but still runs out of
resources while verifying our example in Figure~\ref{Mex1}.

Another way of partitioning an array is to exploit its semantic properties
and split the elements into groups. Grouping array cells of similar properties,
e.g.~\cite{liu2014abstraction}, has the advantage that it allows partitions to
be non-contiguous. This is orthogonal to the work of Dillig et
al.~\cite{dillig2010fluid} where they introduce fluid updates of arrays using
indexed locations along with bracketing constraints, to specify the concrete
elements being updated. Cornish et al.~\cite{cornish2014analyzing} apply a
program-to-program translation over the LLVM intermediate representation,
followed by a scalar analysis. Although the abstraction in these approaches is
expressible as a composition of our abstraction followed by further
abstraction, our implementation of $\fullarrayaccess(\Stmt)$ guarantees an
array-free and loop-free programs whenever possible. Moreover, we exploit the
power of model checking to obtain a precise path-sensitive analysis. 

In a recent work~\cite{monniaux2016cell}, Monniaux et al. proposed to
convert array programs into array-less non-linear Horn clauses. The precision
of this transformation is adjustable through a Galois connection parameterized
by the number of distinguished cells, which, however, needs to be decided
manually. An analogous technique, based on Horn clauses over array variables
that requires user inputs in form of pre- and post-conditions, is that of
\cite{de2015rule} proposed by De Angelis et al. This is practically infeasible
for real-life programs, and therefore affects scalability. Another limitation
of these approaches is their back-end solvers which cannot handle non-linear
arithmetic. This reflects in their inability to work for our motivating example
of Figure~\ref{Mex1}. Besides, unlike these approaches, our technique needs no manual
input and successfully scales to large industry code.
 
Template-based methods~\cite{beyer2007invariant,Gulwani2008Lifting} have been
very useful in synthesizing invariants but these techniques are ultimately
limited by the space of possible templates that must be searched for a good
candidate. This has led to semi-automatic approaches, such
as~\cite{flanagan2002predicate}, where the predicates are usually suggested by
the user. Our approach, in contrast, is fully automatic and proves safety by
solving a bounded model checking instance instead of computing an invariant
explicitly.

Verification tools based on CEGAR have been applied successfully to certain
classes of programs, e.g. device drivers~\cite{ball2002s}. However, this
technique is orthogonal to ours. In fact, a refinement framework in addition to
our abstraction would make our technique complete too. Several other techniques
have been used to scale bounded model checkers to tackle complex, real-world
programs such as acceleration~\cite{kroening2013under} and
loop-abstraction~\cite{darke2015over}. But these techniques are not shown to be
beneficial in abstracting complex data structures.
Booster~\cite{alberti2014booster} integrates acceleration and lazy abstraction
with interpolants for arrays. However, there are syntactic restrictions that
limit the applicability of acceleration in general for programs handling
arrays. Also, interpolation over array properties is difficult, especially
since the goal is not to provide any interpolant, but one that generalizes well
to an invariant~\cite{alberti2014extension,alberti2015polyhedra}.

\section{Conclusions and Future Work}
\label{conclusion}

Verification of programs with loops iterating over arrays is a challenging problem because of large sizes of arrays.
We have explored a middle ground between the two
extremes of relying completely on dynamic approaches of using model checkers on the one hand and using completely
static analysis involving complex domains and fix point computations on the other hand. Our experience shows that using static analysis
to transform the program and letting the model checkers do the rest is a sweet spot that enables verification of
properties of arrays using an automatic technique that is generic, sound, scalable, and reasonably precise.

Our experiments show that the effectiveness of our technique depends on the characteristics of programs and
properties sought to be verified.
We are able to eliminate 40-90\% of false warnings from
diverse applications.
This is a significant value addition to static analysis that try to find
defects and end up generating a large number of warnings which need to be resolved manually for safety critical applications. 
Our effort grew out of our own experience of such manual reviews which showed
a large number of warnings to be false positives.

We plan to make our technique more precise by augmenting it with a refinement step to verify the programs that are 
reported as unsafe by our current technique. 
Finally, we wish to extend our technique on other data structures such as maps or lists. 


\Urlmuskip= 0mu plus 2mu\relax
\bibliographystyle{abbrv}
\bibliography{refs}

\end{document}